\RequirePackage{fix-cm}

\let\origvec\vec

\let\vec\origvec

\documentclass[smallextended]{svjour3}       
\smartqed  
\usepackage{graphicx}

\usepackage{amsmath}

\usepackage{amsthm}

\usepackage{pmat}
\usepackage{subfig}
\usepackage{caption}

\usepackage{amssymb}

\usepackage[retainorgcmds]{IEEEtrantools}

\renewcommand{\theequationdis}{{\normalfont (\theequation)}}
\renewcommand{\theIEEEsubequationdis}{{\normalfont (\theIEEEsubequation)}}

\newtheorem{thm}{Theorem}

\newtheorem*{dfn}{Definition}
\newtheorem*{rmk}{Remark}
\newtheorem{rmk-4}{Remark}
\newtheorem{cl}{Corollary}

\begin{document}

\title{The Interpretation of Linear Prediction by Interpolation Framework and Several following Results}


\titlerunning{The Interpretation of Linear Prediction by Interpolation Framework}

\author{Changcun Huang 
}


\institute{Changcun Huang \at
}


\maketitle

\begin{abstract}
This paper gives a general interpretation of Linear Prediction ($LP$) by interpolation framework different from the perspective of statistics. This interpretation is proved to be useful by several following results, such as: The mechanism of widely used least square estimation of $LP$ coefficients can be explained more intuitively. In data modeling, $LP$ coefficients cannot distinguish signals spanned by the same interpolation bases. Two new general $LP$ constructive methods instead of least square estimation are presented with their upper bounds of approximation error and some properties given; one is based on $DCT\text{-1}$ and the other is based on difference operator. We also establish the relationship between $LP$ and Taylor series.
\end{abstract}

\section{Introduction}
\label{intro}
Linear prediction had long been widely used in speech analysis, geophysics and neurophysics even until 1975 \cite{[1]}. Two-dimensional $LP$ is a fundamental image model \cite{[2]} contributing to many useful image processing algorithms, such as image restoration \cite{[3]}. $LP$ is also one of the most important methods of time series analysis \cite{[4]}.  It is very impressive to see the effect of $LP$'s signal representation in engineering. However, nearly all the theoretical descriptions such as \cite{[5]} of $LP$ are based on Kolmogorov \cite{[6]} or Winer \cite{[7]}'s work, which are both under the background of statistics.

This paper tries to interpret $LP$ by interpolation framework different from the perspective of statistics and is organized as follows: In Section 2, We first give the equivalent analytical form of real signals that $LP$ can represent and then introduce the interpolation framework. Section 3 explains $LP$'s approximation ability to arbitrary signals by least square estimation. Two general constructive methods are presented in Section 4. Finally, the conclusion is in Section 5.

\section{Interpolation Framework}
\label{sec:1}
\subsection{Analytical form of $LP$ represented signals}
\label{sec:2}

\begin{dfn}
We denote a sequence of data points by $f_n$ or $f(n)$, which is also referred to as ``infinite-length discrete signal'' when $n$ is infinite or ``finite-length discrete signal'' otherwise.
\end{dfn}

Textbooks about combinatorial mathematics or time series analysis such as \cite{[8],[9]} usually discuss the homogeneous linear difference equation, in which it shows that if an order $p$ $LP$ recurrence
\begin{IEEEeqnarray}{rCl}
f_n = \sum_{k = 1}^{p}a_kf_{n-k}
\end{IEEEeqnarray}
represents any real signal, the signal must be in the form
\begin{IEEEeqnarray}{rCl}
f_n = \sum_{i = 1}^{p_1}\sum_{k = 0}^{e_i - 1}(b_{i_k}n^{k}\rho_i^ncos(n\theta_i) + c_{i_k}n^{k}\rho_i^nsin(n\theta_i)),
\nonumber \\*
\end{IEEEeqnarray}
where $b_{i_k}$, $c_{i_k}$, $\rho_i$, and $\theta_i$ are real numbers, and $e_i$ is a positive integer.

The choices of parameters $\rho_i$, $\theta_i$ and $e_i$ correspond to different kinds of roots of the characteristic polynomial equation of (1). For instance, if $\theta_i \ne 0$ and $\pi/2$, $e_i = 1$, and $\rho_i = 1$, when $f_n$ is a trigonometric sum, the roots are composed of different complex numbers together with their conjugates; if $\theta_i = 0$ or $\pi/2$, $\rho_i \ne \rho_j \ne 0$ and $e_i = 1$, when $f_n$ is an exponential sum, the characteristic polynomial equation of (1) has $p$ different real roots; $f_n$ can also be a polynomial sequence if $\theta_i = 0$ or $\pi/2$, $\rho_i = 1$ and $e_i > 1$, when the roots are all repeated 1's.

In (2), once the parameters of $\rho_i$, $\theta_i$ and $e_i$ are fixed by certain roots of a characteristic polynomial equation, the corresponding $LP$ recurrence can represent all signals in the form of (2) as $b_{i_k}$ and $c_{i_k}$ arbitrarily change.

By the relationship between the roots number and the degree of a polynomial equation (equals the $LP$ order), no matter what kind of signal (2) is, the order $p$ of (1) is always equal to the number of $n^{k}\rho_i^ncos(n\theta_i)$ and $n^{k}\rho_i^nsin(n\theta_i)$ (with $n$ fixed).

In addition, the converse of above conclusion is also true: Any signal in the form of (2) can be iterated by a $LP$ recurrence. For example, to produce a trigonometric sum (a special case of (2))
\begin{IEEEeqnarray}{rCl}
f_n = \sum_{k = 1}^{p_1}(b_kcos(n\theta_k) + c_ksin(n\theta_k)),
\end{IEEEeqnarray}
the solution of (1) should be $\hat{f}_n = \sum_{k=1}^{p_1}({d_kr_k^n + d_k'\bar{r}_k^n})$,
where $r_k = cos(\theta_k) + isin(\theta_k)$ and $\bar{r}_k$ is the complex conjugate of $r_k$. We just need to let $d_k = \frac{1}{2}b_k - \frac{1}{2}c_ki$ and $d_k' = \bar{d_k}$, then one of solutions of $LP$ recurrence (1) is the given trigonometric sum. Similarly, the $LP$ recurrence of other kinds of signal in the form of (2) can be constructed by the theory of homogeneous linear difference equations.

We combine the above descriptions into a unified conclusion as follows.
\begin{thm}
The order $p$ $LP$ (1) is equivalent to (2) in representing real signals.
\end{thm}

\subsection{Interpolation framework}
Note that (3) is related with the common trigonometric interpolation form \cite{[11]} and the polynomial special case of (2) associates with Lagrange interpolation form, where $f_n$'s can be considered as the discrete points of a continuous function $f(t)$ which is obtained by interpolation methods. It's natural to relate (2) to a more general interpolation form different from the above two special cases, which can be written as
\begin{IEEEeqnarray}{rCl}
f(t) = \sum_{i = 1}^{p_1}\sum_{k = 0}^{e_i - 1}(b_{i_k}n^{k}\rho_i^ncos(\theta_it) + c_{i_k}n^{k}\rho_i^nsin(\theta_it)).
\nonumber \\*
\end{IEEEeqnarray}

In the context of interpolation framework, we can borrow some related concepts or thoughts familiar with us, helping to understand $LP$. For example, given some discrete points, the constructed formula (4) by interpolation methods is usually called ``interplant''. Interplant gives a continuous function; however, in this paper, we only study some discrete points of the interplant, so that somewhat a ``lower level'' concept is needed:

\begin{dfn}
In (2), the arbitrary changes of $b_{i_k}$ and $c_{i_k}$ with fixed $n^{k}\rho_i^ncos(n\theta_i)$ or $n^{k}\rho_i^nsin(n\theta_i)$ form the real solution space of a given $LP$ recurrence. We call $n^{k}\rho^ncos(n\theta)$ or $n^{k}\rho^nsin(n\theta)$ used in the form of (2) the interpolation basis, which can be written in the set form
\begin{IEEEeqnarray}{rCl}
\{{n^{k}\rho^ncos(n\theta), n^{k}\rho^nsin(n\theta)} \mid \theta, \rho \in \mathbb{R}, k \in \mathbb{Z} \text{ } and \text{ } k\ge 0\},
\nonumber \\*
\end{IEEEeqnarray}
where $k$ is a nonnegative integer, and $\theta$, $\rho$ are arbitrary real numbers, and $n$ is considered to be a constant when refers to this concept.

If $f_n$ can be expressed in the form of (2), we say that $f_n$ is spanned by its corresponding interpolation bases.
\end{dfn}

We avoid using the concept ``basis'' in function analysis \cite{[10]} since its rigorous definition is not necessarily needed.

Under the above concept, a $LP$ recurrence is equivalent to a set of interpolation bases, while the weights of the interpolation bases can be obtained from several initial values of the recurrence.

Based on this definition and Theorem 1, we can describe an important property of $LP$'s signal representation ability:
\begin{thm}
$LP$ coefficients cannot distinguish signals spanned by the same set of interpolation bases.
\end{thm}
\begin{proof}
A set of $LP$ coefficients yields a $LP$ recurrence with respect to a set of interpolation bases. A certain set of $LP$ coefficients can represent all the signals spanned by the corresponding set of interpolation bases without the ability to distinguish them.
\end{proof}

\begin{rmk}
$LP$ coefficients related features such as $LPC$ \cite{[12]} and $PLP$ \cite{[13]} are successfully applied in speech analysis, where $LP$ is considered as a basic model or a parameterization method of speech signals for pattern recognition. However, by Theorem 2, $LP$ method may be insufficient sometimes when different segments of the speech signal have the same interpolation bases, which may need to be improved.
\end{rmk}

\subsection{Summary}
We first gave the analytical form of signals iterated by $LP$, and then the interpolation framework was introduced. A useful conclusion about $LP$ coefficients widely used in engineering was discussed in Theorem 2. The left parts of this paper are all under this framework.

\section{Interpreting least square estimation of $LP$ by interpolation framework}
We know that any finite-length discrete signal can be approximated by $LP$ via least square estimation of $LP$ coefficients, which is widely used in engineering \cite{[1]}. The results of Section 2 cannot explain this phenomenon because those signals mentioned above are constrained to certain types. Under interpolation framework, we can give an interpretation here.

\subsection{Automatic selection of interpolation bases}
\begin{thm}
\label{Theorem 2}
For any given finite-length discrete signal $f_n$, the least square estimation of order $p$ $LP$ is equivalent to selecting the best $p$ interpolation bases from the set of (5) to approximate $f_n$ by minimum error. The best selection is unique once the order $p$ is fixed.
\end{thm}
\begin{proof}
Least square estimation provides a unique optimal solution of $LP$ coefficients $a_k$'s as in (1). A unique solution of $a_k$'s determines a unique set of interpolation bases belonging to (5) by solving a difference equation specified by $a_k$'s. Different sets of $LP$ coefficients correspond to different sets of interpolation bases, respectively. The $LP$ order $p$ is the degree of characteristic polynomial equations, so that there will be $p$ interpolation bases.
\end{proof}

\begin{rmk}
Advantages of least square estimation: The most interesting and powerful effect of $LP$'s approximation by least square method is that it can select the best interpolation bases automatically from a wide range bases of (5) via arbitrary changes of parameters $\theta$,  $\rho$ and $k$. The usual methods of interpolation or approximation generally fix the interpolation bases first and the performance is constrained to those fixed bases, such as Lagrange interpolation.
\end{rmk}

\subsection{Interpretations of $LP$ order adjusting methods}
In least square estimation of $LP$ coefficients, the $LP$ order $p$ must be manually determined first. Various methods were developed to choose the best $LP$ order, nearly all of which are based on the fact that increasing the order $p$ within certain range leads to less approximation error \cite{[1],[4]}. The principle underlying this fact can be explained by Theorem 3.

According to Theorem 3, higher $LP$ order results in more interpolation bases, which could improve the approximation performance if data points were not modeled well by less interpolation bases. This is similar to polynomial fitting. The simplest data structure is linear, so degree one polynomial is enough; when the data structure is nonlinear, we must increase the polynomial degree to fit that. The more ``complex'' the data, the more interpolation bases are needed.

\section{Two general $LP$ constructive methods}
Although least square method can automatically select the interpolation bases for arbitrary finite-length discrete signals, it is done implicitly by $LP$ coefficients. Based on the results of Section 2, it's natural to choose the interpolation bases directly to construct $LP$. In what follows, two new general constructive methods will be presented.

There's theoretical significance of the two following methods. They can establish relationships between $LP$ and some other branches of signal processing or mathematics. The ``dense'' property of constructed $LP$ in the whole set of $LP$ may lead to standard methods to study all $LP$, just like using Taylor series to study all kinds of smooth functions.

They also have potential engineering applications. For example, one of the results will tell us that adjusting the $LP$ order is not the only way to improve $LP$'s approximation, while increasing the sampling frequency of data points is also beneficial.

\subsection{$DCT$ method}
In trigonometric sum (3), if the coefficient $c_k$ of $sin(n\theta_k)$ is zero, (3) reduces to a sequence of cosine sum, which is associated with the discrete cosine transform ($DCT$) \cite{[14]}. Because $DCT$ can fit any finite-length discrete signal, a corresponding general $LP$ constructive method follows.

\begin{thm}
\label{Theorem 3}
Any sequence of $f(n)$ for $n = 0, 1, \cdots, N - 1$ can be approximated or parameterized by $LP$ constructed through $DCT\text{-1}$. In constructing process, selecting the $p$ interpolation bases with respect to $p$ largest absolute-value weights gives the minimum upper bound of approximation error than other selections.
\end{thm}
\begin{proof}
There are four common forms of $DCT$, including $DCT\text{-1}$, $DCT\text{-2}$, $DCT\text{-3}$ and $DCT\text{-4}$ \cite{[14]}. Among them, $DCT\text{-1}$ is most appropriate for constructing $LP$ because we can directly make it into the form (3). $DCT\text{-1}$ theory states that for any given $N$ distinct points of $f(n)$, there exists a unique interpolation form of cosine sum such that
\begin{IEEEeqnarray}{rCl}
f(n) = b_{0} + \sum_{k = 1}^{N - 1}b_kcos(n\theta_k),
\end{IEEEeqnarray}
where $b_0$ and $b_k$'s are the $DCT\text{-1}$ of $f(n)$, and $\theta_k = \frac{\pi k}{N -1}$.

If $b_0$ and $b_k$'s are all nonzero, the number of interpolation bases in (6) is $N_I = N$ and $2(N - 1) + 1$ $LP$ coefficients are needed to iterate $f(n)$. It's meaningless to do this when we want to parameterize $f(n)$ by much less $LP$ coefficients. However, we can select some of the interpolation bases to approximate $f(n)$ instead of iterating it. After selecting probable interpolation bases, we can use them to construct $LP$.

We choose $p$ of the $N_I$ interpolation bases (may include the constant term $1$ by letting $\theta_0 = 0$) and let
\begin{IEEEeqnarray}{rCl}
\hat{f}(n) = \sum_{k=1}^{p}{b_{\sigma(k)}\cos(n\theta_{\sigma(k)})},
\end{IEEEeqnarray}	
where $\sigma(k)$'s are the indexes of selected interpolation bases and $ p < N_I$. Then the mean-square approximation error is
\begin{IEEEeqnarray}{l}
E = \frac{1}{N}\sum_{n=0}^{N - 1}{(f(n) - \hat{f}(n))^2}\nonumber \\*
\quad = \frac{1}{N}\sum_{n=0}^{N - 1}\sum_{k=p + 1}^{N_I}{(b_{\sigma{'}(k)}\cos(n\theta_{\sigma{'}(k)}) })^2 \nonumber \\*
\quad \le  \frac{1}{N}\sum_{n=0}^{N - 1}(\sum_{k=p + 1}^{N_I}{b_{\sigma{'}(k)}^2 }) = \sum_{k=p + 1}^{N_I}{b_{\sigma{'}(k)}^2 },
\end{IEEEeqnarray}
where $\sigma{'}(k)$'s are the indexes of left interpolation bases after selection. (8) indicates that the approximation error is bounded by the sum of squares of the weights with respect to the left interpolation bases. So choosing interpolation bases with largest absolute-value weights gives the minimum upper bound of approximating error.

If parts of $b_0$ and $b_k$'s are zero, just analyse the nonzero parts or directly use them to construct $LP$.

The sparsity of $DCT\text{-1}$, i.e., much of $b_k$'s and $b_0$ in (6) are zero or close to zero, has been demonstrated in the case of Markov-1 signal, by which $DCT$ is regarded as a qualified approximation to the $KLT$ optimal transform \cite{[15]}. The sparsity of $DCT$ is also verified in engineering applications and successfully used in image and speech compression \cite{[14],[15]}. It's feasible to choose several interpolation bases with larger absolute value weights to approximate $f(n)$ by certain precision due to this sparsity.

According to (8), if we select more interpolation bases, the upper bound of error also becomes small. More interpolation bases cause higher $LP$ order. This may explain some of the $LP$ order adjusting methods.
\end{proof}

\begin{rmk}
``Dense'' property of the constructed $LP$: We know that the set $\mathbb{Q}$ of all rational numbers is dense in the set $\mathbb{R}$ of all real numbers \cite{[16]}, which means that every real number can be approximated by rational numbers. It's a good way to study the real number by relatively ``simple'' rational number in terms of this dense property. We can generalize this idea to the set of all $LP$, though not necessarily the rigorous definition.

The set of $LP$ constructed by $DCT\text{-1}$ doesn't include the signals spanned by the interpolation bases other than \{$cos(n\theta)$\}. However, because the constructed $LP$ can approximate arbitrary signals by certain precision, we may consider that the set of constructed $LP$ is ``dense'' in the set of all $LP$; i.e., every $LP$ can be approximated by constructed $LP$ to some extent.
\end{rmk}

\subsection{Difference operator method}
The following result is associated with the special case of (2), i.e., the polynomial sequence. We'll use polynomial sequence to measure arbitrary finite-length discrete signal in terms of $LP$ constructed by difference operator. The relationship between $LP$ and Taylor series will be established.
\begin{thm}
Let $f(n)$ for $n = 0, 1,\cdots, N-1$ be equidistant sampled points of smooth function $f(x)$ on any closed interval. Then $f(n)$ can be approximated or parameterized by $LP$ with any precision as $N$ increases in the form of
\begin{IEEEeqnarray}{rCl}
\hat{f}_p(n) = f(n - 1) + \Delta f(n - 1) + \Delta^2f(n - 1) \>+ \nonumber \\*
 \cdots + \Delta^{p - 1}f(n - 1),\IEEEeqnarraynumspace \IEEEeqnarraynumspace \IEEEeqnarraynumspace
\end{IEEEeqnarray}
where $p$ is the $LP$ order and $\Delta^kf(n - 1)$ is the difference operator \cite{[17]}
\begin{IEEEeqnarray}{rCl}
\Delta^kf(n - 1) = \sum_{i = 0}^{k}(-1)^ic_k^if(n - 1 - i),
\end{IEEEeqnarray}
where $k$ is a nonnegative integer and $c_k^i$ is the binomial coefficient. By simple manipulation, (9) can be written as $\hat{f}_p(n) = \sum_{k = 1}^{p}a_kf(n - k)$.
\end{thm}
\begin{proof}
The first part of the proof is totally combinatorial. Substituting (10) into (9) and putting $f(n - 1 -i)$ of the same $i$ with different coefficient together, we have
\begin{IEEEeqnarray}{l}
\hat{f}_p(n) = \sum_{k = 0}^{p - 1}c_k^{0}(-1)^0f(n - 1) + \sum_{k = 1}^{p - 1}c_k^{1}(-1)^1f(n - 2) \>+
\nonumber \\*
\sum_{k = 2}^{p - 1}c_k^{2}(-1)^2f(n - 3) + \cdots + \sum_{k = p - 1}^{p - 1}c_k^{p - 1}(-1)^{p - 1}f(n - p).
\nonumber \\*
\end{IEEEeqnarray}
Because $\sum_{k = 0}^{n}c_k^i= c_{n + 1}^{i + 1}$ \cite{[8]}, (11) can be further simplified into
\begin{IEEEeqnarray}{rCl}
\hat{f}_p(n) = c_p^1(-1)^0f(n - 1) + c_p^2(-1)^1f(n - 2) \>+
\nonumber \\*
c_p^3(-1)^2f(n - 3) + \cdots + c_p^p(-1)^{p - 1}f(n - p).
\end{IEEEeqnarray}
Noting that
\begin{IEEEeqnarray*}{lCl}
f(n) - \hat{f}_p(n)
\nonumber \\*
= f(n) - [c_p^1(-1)^0f(n - 1) + c_p^2(-1)^1f(n - 2) \>+
\nonumber \\* \quad \quad \quad \quad \quad \quad
\cdots + c_p^p(-1)^{p - 1}f(n - p)]
\nonumber \\*
= c_p^0(-1)^0f(n) + c_p^1(-1)^1f(n - 1) + c_p^2(-1)^2f(n - 2)
\nonumber \\* \quad \quad \quad \quad \quad \quad \quad
+\> \cdots + c_p^p(-1)^{p}f(n - p)
\nonumber \\*
= \Delta^{p}f(n),
\end{IEEEeqnarray*}	
so
\begin{IEEEeqnarray}{rCl}
f(n) = \hat{f}_p(n) + \Delta^pf(n).
\end{IEEEeqnarray}

(13) is the key of the proof, which shows that the difference between $f(n)$ and the approximation of order $p$ $LP$ constructed is the order $p$ difference of $f(n)$, i.e., $\Delta^pf(n)$. The difference operator $\Delta^{k}f(n)$ has good properties to polynomial sequences; that is, if $f(n)$ is derived from degree $k - 1$ polynomial, then $\Delta^{k}f(n)$ will be zero \cite{[17]}. The left proof is based on (13) in two cases:

Case 1: Interpolation. Of course, any sequence of $N$ data points can be interpolated by degree $N-1$ polynomial via Lagrange interpolation, when $\Delta^{N}f(n) = 0$; however, this trivial case is not our concern. What we are interested in is using fewer $LP$ coefficients to represent $N$ data points. So only the case that $f(x)$ is a degree $p - 1$ polynomial satisfying $p - 1 < N - 1$ is taken into consideration, when $\Delta^pf(n) = 0$, and therefore $f(n)$ can be iterated by order $p$ $LP$.

Cases 2: Approximation. If $\Delta^pf(n)$ is nonzero for $p - 1 < N - 1$ or $p$ is not small enough for parameterization despite being Case 1, $LP$ can be used to approximate $f(n)$. In this case, $\Delta^pf(n)$ also has some good properties if $f(x)$ is a smooth function on closed interval. The details are as follows.

The mean-square error of approximating $f(n)$ by $\hat{f}_p(n)$ for $ n = 0, 1, \cdots, N-1$ is
\begin{IEEEeqnarray}{rCl}
E = \frac{1}{N}\sum_{n = 0}^{N - 1}(f(n) - \hat{f}_p(n))^2 = \frac{1}{N}\sum_{n = 0}^{N - 1}(\Delta^pf(n))^2.
\nonumber \\*
\end{IEEEeqnarray}
We see
\begin{IEEEeqnarray}{l}
\left|\Delta^pf(n)\right| = \left|\sum_{i = 0}^p(-1)^ic_p^if(n - i)\right|
\nonumber \\* \quad \quad \quad \quad
\le \left|\sum_{\text{i is even}}^pc_p^if_M(n) - \sum_{\text{i is odd}}^pc_p^if_{m}(n)\right|,
\nonumber \\*
\end{IEEEeqnarray}
where $f_M(n)$ and $f_{m}(n)$ are respectively the maximum and minimum value of $f(n)$ within a local neighbourhood at position $n$, i.e., $f(n - i)$ for $i = 0, 1, \cdots, p$. Since $\sum_{i = 0}^p(-1)^ic_p^i = 0$ \cite{[8]}, which can be written as $\sum_{\text{i is even}}^pc_p^i = \sum_{\text{i is odd}}^pc_p^i = \lambda$,
we have
\begin{IEEEeqnarray*}{rCl}
\left| \Delta^pf(n)\right| \le \lambda(f_M(n) - f_m(n)),
\end{IEEEeqnarray*}
where $\lambda$ is a nonnegative constant determined by the order of difference operator. Letting $\omega = \max \limits_{n}\{f_M(n) - f_m(n)\}$,
then
\begin{IEEEeqnarray}{rCl}
\left| \Delta^pf(n)\right| \le \lambda\omega
\end{IEEEeqnarray}
for all $n = 0, 1, \cdots, N-1$. (14) and (16) imply
\begin{IEEEeqnarray}{rCl}
E \le \lambda^2\omega^2.
\end{IEEEeqnarray}

Now discuss (17). Denote the closed interval between $n$ and $n - p$ by $I$. Note that $\omega$ is related to the modulus of continuity of smooth function $f(x)$ on closed interval $I$ \cite{[18]}, which is
\begin{IEEEeqnarray*}{rCl}
\omega_c(\delta; f) = sup\left|f(x_1) - f(x_2)\right|,
\end{IEEEeqnarray*}
where $x_1 \in I$, $x_2 \in I$ and $\left|x_1 - x_2\right| \le \delta$. We know that $\omega_c\to0$ as $\delta\to0$ \cite{[18]}. Therefore, as the length of interval $I$ tends to zero ($\delta \to 0$ simultaneously), the limit of $\omega_c$ is also zero. It's obvious that $\omega \le \omega_c$. Thus, if the sampling frequency is sufficiently high (i.e., $N$ is large enough), the length of $I$ between $n$ and $n - p$ will be as short as possible, resulting in arbitrarily small value of $\omega$. By (17), the approximation error $E$ can also be arbitrarily small. So $LP$ with fixed order can approximate $f(n)$ by any precision as $N$ increases.

As shown in Fig.\ref{Fig.2}, the increasing of $N$ is corresponding to higher sampling frequency to continues function $f(x)$, which makes the neighborhood of $p$ data points more local.
\end{proof}

\begin{rmk-4}
Theorem 5 tells us that even if the $LP$ order is unchanged, we can improve the approximating performance by dealing with the original data such as increasing the sampling frequency.
\end{rmk-4}

\begin{rmk-4}
Obviously, the set of constructed $LP$ by difference operator also has the ``dense'' property in the set of all $LP$.
\end{rmk-4}

\begin{rmk-4}
Distinction of the constructed $LP$'s representation of different signals: Since difference operator is the same for all signals, how could different signals be classified by this kind of constructed $LP$? The local nonlinearity of data points measured by $LP$ order and the initial values of $LP$ recurrence are the distinctive features. Signals derived from different degree polynomials have different $LP$ order. If signals are from the same degree polynomial, the distinction is the set of initial values of $LP$ recurrence. Otherwise, given the same approximation error and sampling frequency, the $LP$ order or initial values will be distinct among different signals.
\end{rmk-4}

\begin{figure}[!t]
\captionsetup{justification=centering}
\centering
\subfloat[Smaller $N$ with larger neighborhood.]{\includegraphics[width=2.0in]{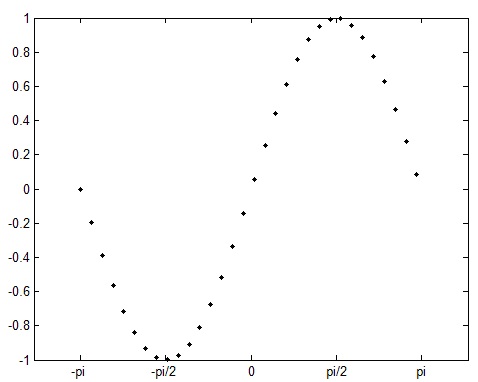}} \quad
\subfloat[Larger $N$ with smaller neighborhood.]{\includegraphics[width=2.0in]{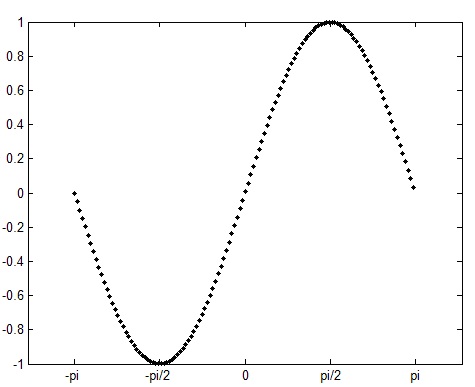}}
\caption{$N$ corresponds to the sampling frequency to continues function.}
\label{Fig.2}
\end{figure}

\subsection{Relationship to Taylor series}

Theorem 5 discussed the case when the order $p$ of $LP$ constructed by difference operator is fixed and the number $N$ of data points is varying. What's the effect if we fix $N$ and change $p$? Taylor series is needed to answer this question.
\begin{cl}
For fixed $N$, the approximation error of $LP$ constructed via difference operator to $f(n)$ is determined by the smoothness of function $f(x)$ measured by the polynomial degree of its Taylor series.
\end{cl}
\begin{proof}
The trivial case is when $f(x)$ is a degree $k$ polynomial. The approximation error of order $p$ $LP$ is zero if $p \ge k + 1$ by (13), which is related to the Lagrange interpolation.

Otherwise, if $f(x)$ can be approximated well by a degree $k$ polynomial in terms of Taylor series on a closed interval, $f(n)$ will be nearly data points of degree $k$ polynomial such that $\Delta^{k + 1}f(n)$ is close to zero. $\Delta^{k + 1}f(n)$ is exactly the order $k + 1$ $LP$'s approximating error by (13). Thus, increasing the degree of Taylor series makes the approximation error smaller to smooth function; correspondingly, increasing of $LP$ order yields more precise approximation to $f(n)$.
\end{proof}

Corollary 1 has demonstrated that the Taylor series of $f(x)$ affects the approximation error of constructed $LP$ to $f(n)$. Furthermore, they are also similar in the approximating form.

We can see from (9) that current value of $f(n)$ is approximated by its nearest point $f(n - 1)$ adding residual from the first-order difference to higher-order difference. If $f(n)$ is a linear sequence, order 2 $LP$ is enough to iterate it. When $f(n)$ is from quadratic function, order 3 $LP$ is required. If $f(n)$ only can be approximated instead of being iterated, higher order $LP$ yields smaller approximation error. Thus, when $f(n)$ is sampled from more complex nonlinear function, we should increase the $LP$ order to iterate or approximate it.

So is Taylor series. When $f(x)$ is a polynomial, the Taylor series of $f(x)$ is exactly itself, just like $LP$'s iteration of $f(n)$. If $f(x)$ is not a polynomial, the residual of approximation of Taylor series is determined by the local smoothness of $f(x)$ measured by the polynomial degree; the higher the polynomial degree, the smaller the approximation error is.

We may call this constructed $LP$ ``the discrete version of Taylor series''. Fig.\ref{Fig.3} is the simplest case of this comparison. In Fig.\ref{Fig.3}, (a) is the approximation to a smooth function $f(x)$ by its differential (first-order Taylor series) and (b) is approximating $f(n)$ by order 2 $LP$, which are very similar to each other.

\begin{figure}[!t]
\captionsetup{justification=centering}
\centering
\subfloat[Approximation to $f(x)$ by Taylor series.]{\includegraphics[width=2.0in]{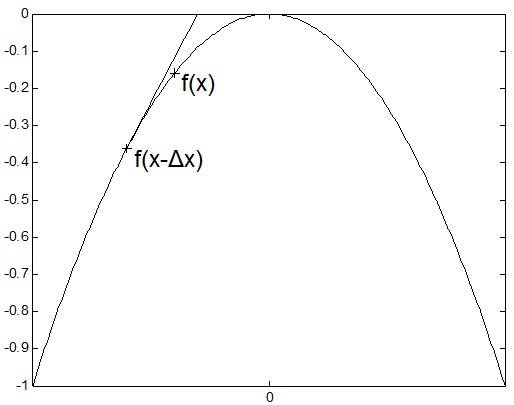}} \quad
\subfloat[Approximation to $f(n)$ by $LP$.]{\includegraphics[width=2.0in]{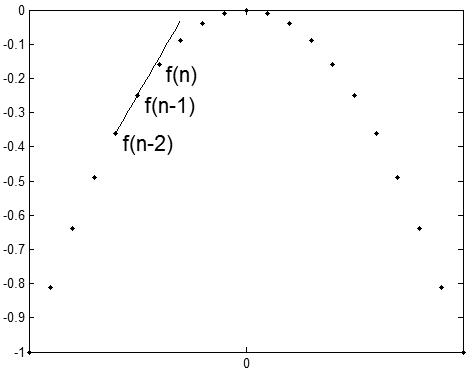}}
\caption{Comparison between Taylor series and $LP$.}
\label{Fig.3}
\end{figure}

\section{Conclusions}
We gave a general interpretation of $LP$ in the interpolation framework as well as several following results that are useful both in engineering and theory. This interpolation framework can help us understand $LP$ from a new viewpoint besides the widely known statistical perspective. We hope that there will be more useful or interesting results discovered underlying this framework.




\end{document}